\documentclass[11pt,a4paper]{llncs}

\usepackage[T1]{fontenc}
\usepackage[latin1]{inputenc}
\usepackage{algorithm}
\usepackage{algorithmic}
\usepackage{float}
\usepackage{graphicx,color}
\usepackage{xspace}
\usepackage{rotating}
\usepackage{geometry}

\newcommand{\Order}[1]{\ensuremath{\mathcal{O}(#1)}}

\newcommand{\nn}[0]{\ensuremath{n_I}}

\newcommand{\SNT}{\mathtt{\symbol{36}}}
\newcommand{\Comment}[1]{{}}

\newcommand{\LCP}[0]{\ensuremath{\mathsf{LCP}}}
\newcommand{\PLCP}[0]{\ensuremath{\mathsf{PLCP}}}

\newcommand{\SUF}[0]{\ensuremath{\mathsf{SA}}}
\newcommand{\LF}[0]{\ensuremath{\mathsf{LF}}}

\newcommand{\lastocc}[0]{\ensuremath{\mathsf{last\_occ}}}

\newcommand{\InverseSUF}[0]{\mathsf{ISA}}

\newcommand{\BWT}{\ensuremath{\mathsf{BWT}}}
\newcommand{\SF}[1]{S_{\SUF[#1]}}

\newcommand{\Keyw}[1]{{\bf #1}}

\newcommand{\ECOM}[1]{\hspace*{0.1cm}%
                     \texttt{\symbol{47}\symbol{42}}~%
                     \mbox{\begin{small}#1\end{small}}~%
                     \texttt{\symbol{42}\symbol{47}}}
\newcommand{\rmq}[0]{\mathsf{RMQ}}

\begin{document}
\title{Lightweight LCP-Array Construction in Linear Time}

\author{Simon Gog and Enno Ohlebusch}
\institute{Institute of Theoretical Computer Science, University of Ulm,
D-89069 Ulm\\
\email{Simon.Gog@uni-ulm.de, Enno.Ohlebusch@uni-ulm.de}}

\maketitle

\begin{abstract}
The suffix tree is a very important data structure in string processing,
but it suffers from a huge space consumption. In large-scale applications, 
compressed suffix trees (CSTs) are therefore used instead. A CST
consists of three (compressed) components: the suffix array, the
$\LCP$-array, and data structures for simulating 
navigational operations on the suffix tree. The $\LCP$-array
stores the lengths of the longest common prefixes of lexicographically 
adjacent suffixes, and it can be computed in linear time.
In this paper, we present new $\LCP$-array construction algorithms
that are fast and very space efficient. In practice, our algorithms
outperform the currently best algorithms.
\end{abstract}

\section{Introduction}

A suffix tree for a string $S$ of length $n$ is a compact trie storing
all the suffixes of $S$. It is an extremely important data structure with
applications in string matching, bioinformatics, and
document retrieval, to mention only a few examples; see e.g.\ \cite{GUS:1997}.
The drawback of suffix trees is their huge space consumption of about 20 times 
the text size, even in carefully engineered implementations.
To reduce their size, several authors provided compressed suffix trees (CSTs);
see e.g.\ \cite{SAD:2007} and \cite{NAV:MAK:2007} for a survey.
A CST of $S$ can be divided into three components: (1) the suffix array
$\SUF$, specifying the lexicographic order of $S$'s suffixes, 
(2) the $\LCP$-array, storing the lengths of
the longest common prefixes of lexicographically adjacent suffixes,
and (3) additional data structures for simulating 
navigational operations on the suffix tree.

Particular emphasis has been put on efficient construction algorithms for all
three components of CSTs. (Here, ``efficiency'' encompasses both construction
\emph{time} and \emph{space}, as the latter can cause a significant memory
bottleneck.) This is especially true for the first component. In the last
decade, much effort has gone into the developement of efficient suffix array 
construction algorithms (SACAs); see \cite{PUG:SMY:TUR:2007} for a survey.
Although linear-time direct SACAs were known since 2003, experiments showed 
\cite{PUG:SMY:TUR:2007} that these were outperformed in practice by SACAs 
having a worst-case time compexity of $O(n^2 \log n)$.
To date, however, the fastest SACA is a linear time 
algorithm \cite{NON:ZHA:CHA:2009}. 
Interestingly, for ASCII alphabet its speed can compete with the 
fastest \LCP-array construction algorithms (LACA) which uses equal or less space.
This is somewhat surprising because sorting all suffixes
seems to be more difficult than computing lcp-values.

As discussed in Section \ref{sec-Related work}, today's best LACAs
\cite{KAR:MAN:PUG:2009,KAS:2001} are linear time algorithms, but they
suffer from a poor locality behavior. In this paper, we present
two very space efficient (using $n$ or $2n$ bytes only) and fast LACAs.
Based on the observation that one cache miss takes approximately the time 
of $20$ character comparisons, we try to trade character comparisons for 
cache misses. The algorithms use the text (string) $S$, the suffix array, 
and the Burrows-Wheeler Transform ($\BWT$). Since most CSAs are based 
on the $\BWT$ anyway, we basically get it for free.
Our experiments show the significance of the algorithms.   
More precisely,
experimental results show that our algorithms outperform
state-of-the-art algorithms \cite{KAR:MAN:PUG:2009,KAS:2001}.
For large texts they are always faster
than the previously best algorithms
. The superiority of our new LACAs
varies with the text size (the larger the better), the alphabet size
(the smaller the better), the number of ``large''
values in the $\LCP$-array (the less the better), and
the runs in the $\BWT$ (the more the better). 
The algorithms work particularly well on two types of data that are
of utmost importance in practice: long DNA sequences
(small alphabet size) and large collections of XML documents 
(long runs in the BWT). 

\section{Related work}
\label{sec-Related work}
In their seminal paper \cite{MAN:MYE:1993}, Manber and Myers did not only
introduce the suffix array but also the longest-common-prefix ($\LCP$)
array. They showed
that both the suffix array and the $\LCP$-array can be constructed
in $O(n \log n)$ time for a string of length $n$. Kasai et al.\ \cite{KAS:2001}
gave the first linear time algorithm for the computation of the $\LCP$-array.
Their algorithm uses the string $S$, the suffix array, the inverse
suffix array, and of course the $\LCP$-array. 
Each of the arrays requires $4n$ bytes (under the assumption that $n<2^{32}$),
thus the algorithms needs $13n$ bytes in total (for an ASCII alphabet).
The main advantage of their algorithm is that it is simple and uses at most
$2n$ character comparisons. But its poor locality behavior results in many 
cache misses, which is a severe disadvantage on current computer architectures.
Manzini \cite{MAN:2004} reduced the space occupancy of Kasai et al.'s 
algorithm to $9n$ bytes with a slow down of about $5\% -10\%$.
He also proposed an even more space-efficient (but slower)
algorithm that overwrites the suffix array. Recently,
K\"arkk\"ainen et al.\ \cite{KAR:MAN:PUG:2009} proposed another variant
of Kasai et al.'s algorithm, which computes a permuted $\LCP$-array 
($\PLCP$-array). In the $\PLCP$-array, the lcp-values are in text order
(position order) rather than in suffix array order (lexicographic order).
This algorithm takes only $5n$ bytes and is much faster than Kasai et al.'s 
algorithm because it has a much better locality behavior. However,
in virtually all applications lcp-values are required to be in suffix array 
order, so that in a final step the $\PLCP$-array must be converted into
the $\LCP$-array. Although this final step suffers (again) from a poor 
locality behavior, the overall algorithm is still faster than 
Kasai et al.'s.
In a different approach, Puglisi and Turpin \cite{PUG:TUR:2008} tried
to avoid cache misses by using the difference
cover method of K\"arkk\"ainen and Sanders \cite{KAR:SAN:2003}.
The worst case time complexity of their algorithm is $\Order{nv}$
and the space requirement is $n+\Order{n/\sqrt{v}+v}$ bytes, where
$v$ is the size of the difference cover. Experiments 
showed that the best run-time is achieved for $v=64$, but the algorithm
is still slower than Kasai et al.'s. This is because it uses
constant time range minimum queries, which take considerable time in practice.
To sum up, the currently best LACA is that of
K\"arkk\"ainen et al.\ \cite{KAR:MAN:PUG:2009}.

\section{Preliminaries}
\label{sec-Preliminaries}
Let $\Sigma$ be an ordered alphabet whose smallest element is the
so-called sentinel character $\SNT$. If $\Sigma$ consists of
$\sigma$ characters and is fixed, then we may view $\Sigma$ as an array
of size $\sigma$ such that the characters appear in ascending
order in the array $\Sigma[0..\sigma -1]$, i.e., 
$\Sigma[0] =\SNT < \Sigma[1] < \dots < \Sigma[\sigma -1]$. In the following,
$S$ is a string of length $n$ over $\Sigma$ having the sentinel character
at the end (and nowhere else). For $0 \leq i \leq n-1$, 
$S[i]$ denotes the \emph{character at position} $i$ in $S$.
For $i \leq j$, $S[i..j]$ denotes the \emph{substring} of $S$ starting with 
the character at position $i$ and ending with the character at position $j$.
Furthermore, $S_i$ denotes the suffix $S[i..n-1]$ of $S$.
The \emph{suffix array} $\SUF$ of the string $S$
is an array of integers in the range $0$ to $n-1$ specifying the
lexicographic ordering of the $n$ suffixes of the string $S$,
that is, it satisfies $\SF{0} < \SF{1} < \cdots < \SF{n-1}$;
see Fig.\ \ref{fig:suffix array} for an example.
In the following, $\InverseSUF$ denotes the inverse of the permutation $\SUF$.

The $\LCP$-array is an array containing the lengths of the 
longest common prefix between every pair of consecutive suffixes in $\SUF$.
We use $lcp(u,v)$ to denote the length of the longest common prefix between 
strings $u$ and $v$. Thus, the lcp-array is an array of integers in the range 
$0$ to $n$ such that $\LCP[0] = -1$, $\LCP[n] = -1$, and
$\LCP[i] = lcp(\SF{i-1},\SF{i})$ for $1\leq i \leq n-1$;
see Fig.\ \ref{fig:suffix array}. For $i<j$, a range minimum query $\rmq(i,j)$
on the interval $[i..j]$ in the $\LCP$-array
returns an index $k$ such that $\LCP[k]=\min \{\LCP[l]\mid i \le l\le j\}$.
It is not difficult to show that $lcp(\SF{i},\SF{j}) = \LCP[\rmq(i+1,j)]$.

The Burrows and Wheeler transform \cite{BUR:WHE:1994} converts a string $S$ 
into the string $\BWT[0..n-1]$ defined by $\BWT[i]=S[\SUF[i] -1]$ for all 
$i$ with $\SUF[i] \neq 0$ and $\BWT[i] = \SNT$ otherwise;
see Fig.\ \ref{fig:suffix array}.
The $\LF$-mapping is defined by $\LF[i] = \InverseSUF[\SUF[i]-1]$ for all 
$i$ with $\SUF[i] \neq 0$ and $\LF[i] = 0$ otherwise;
see Fig.\ \ref{fig:suffix array}. Its long name 
\emph{last-to-first column mapping} stems from the fact that it maps
the last column $L = \BWT$ to the first column $F$, where $F$ contains the 
first character of the suffixes in the suffix array, i.e., $F[i] = S[\SUF[i]]$.
More precisely, if $\BWT[i]=c$ is the $k$-th occurrence of character
$c$ in $\BWT$, then $j = \LF[i]$ is the index such
that $F[j]$ is the $k$-th occurrence of $c$ in $F$.
The $\LF$-mapping can be implemented by
$\LF[i] = C[c]+ occ(c,i)$, where $c = \BWT[i]$,
$C[c]$ is the overall number (of occurrences)
of characters in $S$ which are strictly smaller than $c$, and
$occ(c,i)$ is the number of occurrences of the character 
$c$ in $\BWT[1..i]$. 

\begin{figure}[t]
\begin{footnotesize}
\begin{minipage}[b]{0.5\linewidth}
\centering
\begin{tabular}{rrrrclp{1cm}}
$i$ & $\SUF$ & $\LCP$& $\LF$ & $\BWT$ &$S$                \\
\tt{0} & \tt{18} & \tt{-1} & \tt{15} &n& \tt{\$}                   \\
\tt{1} & \tt{2} & \tt{0} & \tt{11} &l& \tt{\_anele\_lepanelen\$}  \\
\tt{2} & \tt{8} & \tt{1} & \tt{5} &e& \tt{\_lepanelen\$}          \\
\tt{3} & \tt{3} & \tt{0} & \tt{1} &\_& \tt{anele\_lepanelen\$}    \\
\tt{4} & \tt{12} & \tt{5} & \tt{18} &p& \tt{anelen\$}             \\
\tt{5} & \tt{7} & \tt{0} & \tt{12} &l& \tt{e\_lepanelen\$}        \\
\tt{6} & \tt{0} & \tt{1} & \tt{0} &\$& \tt{el\_anele\_lepanelen\$}\\
\tt{7} & \tt{5} & \tt{2} & \tt{16} &n& \tt{ele\_lepanelen\$}      \\
\tt{8} & \tt{14} & \tt{3} & \tt{17} &n& \tt{elen\$}               \\
\tt{9} & \tt{16} & \tt{1} & \tt{13} &l& \tt{en\$}                 \\
\tt{10} & \tt{10} & \tt{1} & \tt{14} &l& \tt{epanelen\$}          \\
\tt{11} & \tt{1} & \tt{0} & \tt{6} &e& \tt{l\_anele\_lepanelen\$} \\
\tt{12} & \tt{6} & \tt{1} & \tt{7} &e& \tt{le\_lepanelen\$}       \\
\tt{13} & \tt{15} & \tt{2} & \tt{8} &e& \tt{len\$}                \\
\tt{14} & \tt{9} & \tt{2} & \tt{2} &\_& \tt{lepanelen\$}          \\
\tt{15} & \tt{17} & \tt{0} & \tt{9} &e& \tt{n\$}                  \\
\tt{16} & \tt{4} & \tt{1} & \tt{3} &a& \tt{nele\_lepanelen\$}     \\
\tt{17} & \tt{13} & \tt{4} & \tt{4} &a& \tt{nelen\$}              \\
\tt{18} & \tt{11} & \tt{0} & \tt{10} &e& \tt{panelen\$}           \\ 
\end{tabular}
\end{minipage}
\hspace{1.2cm}
\begin{minipage}[b]{0.3\linewidth}
\centering
\begin{tabular}{rrrrrrrrrrrr}
$i/j$\ \ &     1&     2&     3&     4&     5&     6&     7&     8&     9&    10&    11\\
\tt{0}\ \  &\tt{\bf{-1}}&\tt{\bf{-1}}&\tt{\bf{-1}}&\tt{\bf{-1}}&\tt{\bf{-1}}&\tt{\bf{-1}}&\tt{\bf{-1}}&\tt{\bf{-1}}&\tt{\bf{-1}}&\tt{\bf{-1}}&\tt{\bf{-1}}\\	\hline
\tt{1}\ \  &    &\tt{\bf{0}}&\tt{\bf{0}}&\tt{\bf{0}}&\tt{\bf{0}}&\tt{\bf{0}}&\tt{\bf{0}}&\tt{\bf{0}}&\tt{\bf{0}}&\tt{\bf{0}}&\tt{\bf{0}}\\	
\tt{2}\ \  &   &   &\tt{\bf{1}}&\tt{\bf{1}}&\tt{\bf{1}}&\tt{\bf{1}}&\tt{\bf{1}}&\tt{\bf{1}}&\tt{\bf{1}}&\tt{\bf{1}}&\tt{\bf{1}}\\	\hline
\tt{3}\ \  &   &   &   &\tt{\bf{0}}&\tt{\bf{0}}&\tt{\bf{0}}&\tt{\bf{0}}&\tt{\bf{0}}&\tt{\bf{0}}&\tt{\bf{0}}&\tt{\bf{0}}\\	
\tt{4}\ \  &   &   &   &    &\tt{\bf{5}}&\tt{\bf{5}}&\tt{\bf{5}}&\tt{\bf{5}}&\tt{\bf{5}}&\tt{\bf{5}}&\tt{\bf{5}}\\	\hline
\tt{5}\ \  &   &   &\tt{\bf{0}}&\tt{\bf{0}}&\tt{\bf{0}}&\tt{\bf{0}}&\tt{\bf{0}}&\tt{\bf{0}}&\tt{\bf{0}}&\tt{\bf{0}}&\tt{\bf{0}}\\	
\tt{6}\ \  &   &   &   &    &  &  &\tt{\bf{1}}&\tt{\bf{1}}&\tt{\bf{1}}&\tt{\bf{1}}&\tt{\bf{1}}\\	
\tt{7}\ \  &   &   &   &    &  &  &  &\tt{\bf{2}}&\tt{\bf{2}}&\tt{\bf{2}}&\tt{\bf{2}}\\	
\tt{8}\ \  &   &   &   &    &  &  &  &  &\tt{\bf{3}}&\tt{\bf{3}}&\tt{\bf{3}}\\	
\tt{9}\ \  &   &   &   &    &  &  &  &  &  &\tt{\bf{1}}&\tt{\bf{1}}\\	
\tt{10}\ \ &   &   &   &    &  &  &  &  &  &  &\tt{\bf{1}}\\    \hline	
\tt{11}\ \ &   &\tt{\bf{0}}&\tt{\bf{0}}&\tt{\bf{0}}&\tt{\bf{0}}&\tt{\bf{0}}&\tt{\bf{0}}&\tt{\bf{0}}&\tt{\bf{0}}&\tt{\bf{0}}&\tt{\bf{0}}\\	
\tt{12}\ \ &   &   &   &    &  &\tt{\bf{1}}&\tt{\bf{1}}&\tt{\bf{1}}&\tt{\bf{1}}&\tt{\bf{1}}&\tt{\bf{1}}\\	
\tt{13}\ \ &   &   &   &    &  &  &  &  &  &\tt{\bf{2}}&\tt{\bf{2}}\\	
\tt{14}\ \ &   &   &   &    &  &  &  &  &  &  &\tt{\bf{2}}\\	\hline
\tt{15}\ \ &\tt{\bf{0}}&\tt{\bf{0}}&\tt{\bf{0}}&\tt{\bf{0}}&\tt{\bf{0}}&\tt{\bf{0}}&\tt{\bf{0}}&\tt{\bf{0}}&\tt{\bf{0}}&\tt{\bf{0}}&\tt{\bf{0}}\\	
\tt{16}\ \ &   &   &   &    &  &  &  &\tt{\bf{1}}&\tt{\bf{1}}&\tt{\bf{1}}&\tt{\bf{1}}\\	
\tt{17}\ \ &   &   &   &    &  &  &  &  &\tt{\bf{4}}&\tt{\bf{4}}&\tt{\bf{4}}\\	\hline
\tt{18}\ \ &   &   &   &    &\tt{\bf{0}}&\tt{\bf{0}}&\tt{\bf{0}}&\tt{\bf{0}}&\tt{\bf{0}}&\tt{\bf{0}}&\tt{\bf{0}}\\	
\end{tabular}
\end{minipage}
\caption{Left: Suffix array, $\LCP$-array, $\LF$-mapping, and $\BWT$ of the 
string $S = \tt{el\_anele\_lepanelen\$}$.
Right: The $\LCP$-array after the $j$th iteration of 
Algorithm \ref{alg-first phase} (omitted entries are not computed yet).
\label{fig:suffix array}}
\end{footnotesize}
\end{figure}

\section{First algorithm}
In this section, we present our first LACA. A pseudo-code description
can be found in Algorithm \ref{alg-first phase} and an application
of it is illustrated in Fig.\ \ref{fig:suffix array}. Furthermore,
Theorem \ref{alg-first phase} does not only prove its correctness but also
explains it. The algorithm is based on Lemma \ref{lem-LCP[LF[i]]}, which
in turn requires the following definition.

Define a function $prev$ by
\[
prev(i) = \max \{j \mid 0\leq j < i \mbox{ and } \BWT[j] = \BWT[i]\}
\]
where $prev(i) = -1$ if the maximum is taken over an empty set. Intuitively,
if we start at index $i$ and scan the $\BWT$ upward, then
$prev(i)$ is the first index at which the same character $\BWT[i]$ occurs.

\begin{lemma}
\label{lem-LCP[LF[i]]}
\[
\LCP[\LF[i]] =\left\{\begin{array}{l}
0 \mbox{, if } prev(i) = -1\\
1+\LCP[\rmq(prev(i)+1,i)] \mbox{, otherwise}
\end{array}\right.
\]
\end{lemma}
\begin{proof}
If $prev(i) = -1$, then $S_{\LF[i]}$ is the lexicographically smallest
suffix among all suffix having $\BWT[i]$ as first character. Hence 
$\LCP[\LF[i]] = 0$. Otherwise, $\LF[prev(i)] = \LF[i]-1$. 
In this case, it follows that
\begin{eqnarray*}
\LCP[\LF[i]] &=& lcp(S_{\SUF[\LF[i]-1]},S_{\SUF[\LF[i]]})
= lcp(S_{\SUF[\LF[prev(i)]]},S_{\SUF[\LF[i]]})\\
&=& 1 + lcp(S_{\SUF[prev(i)]},S_{\SUF[i]}) 
= 1+ \LCP[\rmq(prev(i)+1,i)]
\end{eqnarray*}
\end{proof}

\begin{algorithm}[t]
\caption{Construction of the $\LCP$-array.
\label{alg-first phase}}
\begin{tabbing}
\quad \=\quad \=\quad \=\quad \=\quad \=\quad \=\quad \=\quad \=\quad \=\quad \=\quad \=\quad \kill
01 \> \> $\lastocc[0..\sigma-1] \leftarrow [-1,-1,\ldots,-1]$ \\
02 \> \> $\LCP[0] \leftarrow -1$; $\LCP[n] \leftarrow -1$;
$\LCP[\LF[0]] \leftarrow 0$\\
03 \> \> \Keyw{for} $i\leftarrow 1$ \Keyw{to} $n-1$ \Keyw{do} \\
04 \> \> \>\Keyw{if} $\LCP[i] = \bot$ \Keyw{then} \ECOM{$\LCP[i]$ is undefined} \\
05 \> \> \> \>$\ell \leftarrow 0$\\
06 \> \> \> \>\Keyw{if} $\LF[i]<i$ \Keyw{then} \\   
07 \> \> \> \> \>$\ell \leftarrow \max\{\LCP[\LF[i]]-1,0\}$\\ 
08 \> \> \> \> \>\Keyw{if} $\BWT[i]=\BWT[i-1]$ \Keyw{then} \\   
09 \> \> \> \> \> \> continue at line 12\\   
10 \> \> \> \>\Keyw{while} $S[\SUF[i]+\ell]=S[\SUF[i-1]+\ell]$ \Keyw{do} \\
11 \> \> \> \> \> $\ell\leftarrow\ell+1$ \\
12 \> \> \> \> $\LCP[i] \leftarrow \ell$   \\
13 \> \> \> \Keyw{if} $\LF[i] > i$ \Keyw{then} \\
14 \> \> \> \> $\LCP[\LF[i]] \leftarrow \LCP[\rmq(\lastocc[\BWT[i]]+1,i)]+1$ \\
15 \> \> \> $\lastocc[\BWT[i]] \leftarrow i$ 
\end{tabbing}
\end{algorithm}

\begin{theorem}
Algorithm \ref{alg-first phase} correctly computes the $\LCP$-array.
\end{theorem}
\begin{proof}
Under the assumption that all entries in the LCP-array in the first
$i-1$ iterations of the for-loop have been computed
correctly, we consider the $i$-th iteration and prove:
\begin{enumerate}
\item If $\LCP[i] = \bot$, then the entry $\LCP[i]$ will be computed correctly.
\item If $\LF[i] > i$, then the entry $\LCP[\LF[i]]$ will be computed correctly.
\end{enumerate}

(1) If the if-condition in line 6 is not true, then
$S_{\SUF[i-1]}$ and $S_{\SUF[i]}$ are compared character by character 
(lines 10-11) and $\LCP[i]$ is assigned the correct value in line 12.
Otherwise, if the if-condition in line 6 is true, then
$\ell$ is set to $\max\{\LCP[\LF[i]]-1,0\}$. We claim that
$\ell \leq \LCP[i]$. This is certainly true if $\ell=0$, so suppose that
$\ell = \LCP[\LF[i]]-1 > 0$. According to (the proof of) 
Lemma \ref{lem-LCP[LF[i]]},
$\LCP[\LF[i]] - 1 = lcp(S_{\SUF[prev(i)]},S_{\SUF[i]})$. Obviously,
$lcp(S_{\SUF[prev(i)]},S_{\SUF[i]}) \leq lcp(S_{\SUF[i-1]},S_{\SUF[i]})$,
so the claim follows. 

Now, if $\BWT[i]\neq \BWT[i-1]$, then $S_{\SUF[i-1]}$ and $S_{\SUF[i]}$ 
are compared character by character (lines 10-11),
but the first $\ell$ characters are skipped because they are identical.
Again, $\LCP[i]$ is assigned the correct value in line 12.
Finally, if $\BWT[i] = \BWT[i-1]$, then $prev(i) = i-1$.
This, in conjunction with Lemma \ref{lem-LCP[LF[i]]}, yields
$\LCP[\LF[i]] - 1 = lcp(S_{\SUF[prev(i)]},S_{\SUF[i]})
= lcp(S_{\SUF[i-1]},S_{\SUF[i]}) = \LCP[i]$. Thus, $\ell = \LCP[\LF[i]] - 1$
is already the correct value of $\LCP[i]$. So lines 10-11
can be skipped and the assignment in line 12 is correct.

(2) In the linear scan of the $\LCP$-array, we always have
$\lastocc[\BWT[i]] = prev(i)$. Therefore, it is a direct consequence of
Lemma \ref{lem-LCP[LF[i]]} that the assignment in line 14 is correct.
\end{proof}

We still have to explain how the index $j = \rmq(\lastocc[\BWT[i]]+1,i))$ 
and the lcp-value $\LCP[j]$ in line 14 can be computed efficiently.
To this end, we use a stack $K$ of size $\Order{\sigma}$.
Each element on the stack is a pair consisting of an index and an 
lcp-value. We first push $(0,-1)$ onto the initially 
empty stack $K$. It is an invariant of the for-loop that the stack elements
are strictly increasing in both components (from bottom to top).
In the $i$th iteration of the for-loop, 
before line $13$, we update the stack $K$ by removing all elements 
whose lcp-value is greater than or equal to $\LCP[i]$.
Then, we push the pair $(i, \LCP[i])$ onto $K$. Clearly,
this maintains the invariant. Let $x = \lastocc[\BWT[i]]+1$.
The answer to $\rmq(x,i)$ is the pair $(j,\ell)$ where $j$ is the minimum
of all indices that are greater than or equal to $x$.
This pair can be found by an inspection of the stack.
Moreover, the lcp-value $\LCP[x]+1$ we are looking for is $\ell+1$. 
To meet the $\Order{\sigma}$ space condition of the stack,
we check after each $\sigma$th update if the
size $s$ of $K$ is greater than $\sigma$. If so, 
we can remove $s-\sigma$ elements from $K$
because there are at most $\sigma$ possible queries.
With this strategy, the stack size never exceeds $2\sigma$
and the amortized time for the updates is $\Order{n}$.
Furthermore, an inspection of the stack takes $\Order{\sigma}$ time.
In practice, this works particularly well when there is a run in the $\BWT$ 
because then the element we are searching for is on top of the stack.

Algorithm \ref{alg-first phase} has a quadratic run time in the worst case,
consider e.g.\ the string $S=ababab...ab\SNT$.

At first glance, Algorithm \ref{alg-first phase} does not have any
advantage over Kasai et al.'s algorithm because it holds
$S$, $\SUF$, $\LF$, $\BWT$, and $\LCP$ in main memory.
A closer look, however, 
reveals that the arrays $\SUF$, $\LF$, and $\BWT$ are accessed sequentially 
in the for-loop. So they can be streamed from disk. We cannot avoid
the random access to $S$, but that to $\LCP$ as we shall show next.

Most problematic are the ``jumps'' upwards (line 7 when $\LF[i]<i$) 
and  downwards (line 14 when $\LF[i]>i$).
The key idea is to buffer lcp-values in queues (FIFO data structures)
and to retrieve them when needed.

First, one can show that the condition $\LCP[i]=\bot$ in line 4
is equivalent to 
$i\geq C[F[i]]+occ(\BWT[i],i)$. The last condition can be evaluated 
in constant time and space (increment a counter $cnt(\BWT[i])$
in iteration $i$), so it can replace $\LCP[i]=\bot$ in line 4. This is 
important because in case $j=\LF[i]>i$, the value $\LCP[j]$ is still
in one of the queues and has \emph{not} yet been written to the $\LCP$-array.
In other words, when we reach index $j$, we still have $\LCP[j]=\bot$ although
$\LCP[j]$ has already been computed. Thus,
by the test $i\geq C[F[j]]+occ(\BWT[j],i)$ we can decide 
whether $\LCP[j]$ has already been computed or not.

Second, $\LF[i]$ lies in between
$C[\BWT[i]]$ and $C[\BWT[i]]+occ(\BWT[i],n-1)$, the interval of all
suffixes that start with character $\BWT[i]$. Note that there are 
at most $\sigma$ different such intervals.
We exploit this fact in the following way.
For each character $c\in \Sigma$ we use a queue $Q_c$. 
During the for-loop we add (enqueue) the values 
$\LCP[C[c]], \LCP[C[c]+1],\ldots, \LCP[C[c]+occ(c,n-1)]$ in exactly
this order to $Q_c$. In iteration $i$, an operation $enqueue(Q_c,x)$ 
is done for $c=\BWT[i]$ and $x=\LCP[\rmq(\lastocc[\BWT[i]]+1,i)]+1$
in line $14$ provided that $\LF[i]>i$, and in line $12$
for $c=F[i]$ and $x=\ell$. 
Also in iteration $i$, an operation $dequeue(Q_c)$ is done for $c=\BWT[i]$ 
in line 7 provided that $\LF[i]<i$. This dequeue operation returns
the value $\LCP[\LF[i]]$ which is needed in line 7.
Moreover, if $i < C[F[i]]+occ(\BWT[i],i)$,
then we know that $\LCP[i]$ has been computed previously but is still
in one of the queues. Thus, an operation $dequeue(Q_c)$ is done for
$c=F[i]$ immediately before line $13$, and it returns the value $\LCP[i]$.

The space used by the algorithm now only depends on the size of the 
queues. We use constant size buffers for the queues and 
read/write the elements to/from disk if the bufferes
are full/empty (this even allows to answer an $\rmq$ by binary search in
$\Order{\log(\sigma})$ time). Therefore, only the text $S$ remains in main
memory and we obtain an $n$ bytes semi-external algorithm. 

\section{Improved algorithm}

Our experiments showed that even a careful engineered version of
Algorithm \ref{alg-first phase} does not always beat the currently fastest 
LACA \cite{KAR:MAN:PUG:2009}.
For this reason, we will now present another algorithm that
uses a modification of Algorithm \ref{alg-first phase} in its first phase.
This modified version computes each $\LCP$-entry whose value is smaller 
than or equal to $m$, where $m$ is a user-defined value.
(All we know about the other entries is that they are greater than $m$.)
It can be obtained from Algorithm \ref{alg-first phase}
by modifying lines $8$, $10$, and $14$ as follows:
\begin{itemize}
\item[$08$] \Keyw{if} $\BWT[i]=\BWT[i-1]$ \Keyw{and} $\ell<m$ \Keyw{then}
\item[$10$] \Keyw{while} $S[\SUF[i]+\ell]=S[\SUF[i-1]+\ell]$ \Keyw{and}
$\ell<m+1$ \Keyw{do} 
\item[$14$] $\LCP[\LF[i]] \leftarrow \min 
\{\LCP[\rmq(\lastocc[\BWT[i]]+1,i)]+1,m+1\}$
\end{itemize}
In practice, $m=254$ is a good choice because $\LCP$-values
greater than $m$ can be marked by the value $255$ and each $\LCP$-entry 
occupies only one byte. Because the string $S$ must also be kept in
main memory, this results in a total space consumption of $2n$ bytes.

Let $I=[i \mid 0\leq i < n \mbox{ and } \LCP[i]\geq m]$ be an array 
containing the indices at which the values in the $\LCP$-array
are $\geq m$ after phase 1. In the second phase we have to calculate the 
remaining $\nn = |I|$ many $\LCP$-entries, and we use
Algorithm \ref{alg-second phase} for this task. In essence,
this algorithm is a combination of two algorithms
presented in \cite{KAR:MAN:PUG:2009} that compute the $\PLCP$-array: 
(a) the linear time $\Phi$-algorithm and (b) the $O(n\log n)$ time 
algorithm based on the concept of irreducible lcp-values.
Let us recapitulate necessary definitions from \cite{KAR:MAN:PUG:2009}.

\begin{definition}
For all $i$ with $1\leq i \leq n-1$ let $\Phi[\SUF[i]] = \SUF[i-1]$, and
for all $j$ with $0\leq j \leq n-1$ let $\PLCP[j] = lcp(S_j,S_{\Phi[j]})$.
An entry $\PLCP[j]$, where $j>0$, is called \emph{reducible} if 
$S[j-1] = S[\Phi[j]-1]$; otherwise it is \emph{irreducible}. 
\end{definition}

Note that $\PLCP[\SUF[i]]$ is reducible if and only if $\BWT[i] = \BWT[i-1]$.
This is because $\BWT[i] = S[\SUF[i]-1]$ and
$\BWT[i-1] = S[\SUF[i-1]-1] = S[\Phi[\SUF[i]]-1]$.

\begin{lemma}
\label{lem-PLCP entries}
For all $j$ with $0\leq j \leq n-1$, we have $\PLCP[j] \geq \PLCP[j-1] -1$.
Moreover, if $\PLCP[j]$ is reducible, then $\PLCP[j] = \PLCP[j-1] -1$.
\end{lemma}
\begin{proof}
See \cite{KAS:2001,MAN:2004,KAR:MAN:PUG:2009}.
\end{proof}

The preceding lemma has the following two consequences:
\begin{itemize}
\item If we compute
an entry $\PLCP[j]$ (where $j$ varies from $1$ to $n-1$), then $\PLCP[j-1]$
many character comparisons can be skipped. This is the reason for the
linear run time of Algorithm \ref{alg-second phase}; cf.\ 
\cite{KAS:2001,KAR:MAN:PUG:2009}.
\item If we know that
$\PLCP[j]$ is reducible, then no further character comparison is needed to
determine its value. At first glance this seems to be unimportant because
the next character comparison will yield a mismatch anyway. At second glance,
however, it turns out to be important because the character comparison may
result in a chache miss! (Note that in contrast to the $O(n\log n)$ time 
algorithm in \cite{KAR:MAN:PUG:2009}, the $\Phi$-algorithm 
does not make use of this property.)
\end{itemize}

\begin{algorithm}
\caption{Phase 2 of the construction of the $\LCP$-array. 
(In practice $\SUF[n\!-\!1]$ can be used for the undefined value $\bot$ because
the entries in the $\Phi$-array are of the form $\SUF[i-1]$, i.e., 
$\SUF[n\!-\!1]$ does not occur in the $\Phi$-array.)
\label{alg-second phase}}
\begin{tabbing}
\quad \=\quad \=\quad \=\quad \=\quad \=\quad \=\quad \=\quad \=\quad \=\quad \=\quad \=\quad \kill
01 \> \> $b[0..n-1] \leftarrow [0,0,\ldots,0]$ \\
02 \> \> \Keyw{for} $i\leftarrow 0$ \Keyw{to} $n-1$ \Keyw{do} \\
03 \> \> \> \Keyw{if} $\LCP[i] > m$ \Keyw{then}  \\
04 \> \> \> \>$b[\SUF[i]] \leftarrow 1$ \ECOM{the $b$-array can be computed in
  phase 1 already}\\
05 \> \> $\Phi[0..\nn-1] \leftarrow [\bot,\bot,\ldots,\bot]$ \\    
06 \> \> initialize a rank data structure for $b$ \\
07 \> \> \Keyw{for} $i\leftarrow 0$ \Keyw{to} $n-1$ \Keyw{do}\ECOM{stream 
$\SUF$, $\LCP$, and $\BWT$ from disk} \\
08 \> \> \> \> \Keyw{if} $\LCP[i]>m$ \Keyw{and} $\BWT[i]\not=\BWT[i\!-\!1]$
\Keyw{then} \ECOM{$\PLCP[\SUF[i]]$ is irreducible}\\ 
09 \> \> \> \> \> $\Phi[rank_1(\SUF[i])] \leftarrow \SUF[i-1]$  \\
\\
10 \> \> $j_I \leftarrow 0$ \\
11 \> \> $\ell \leftarrow m+1$ \\
12 \> \> $\PLCP[0..\nn-1] \leftarrow [0,0,\ldots,0]$ \\ 
13 \> \> \Keyw{for} $j\leftarrow 0$ \Keyw{to} $n-1$ \Keyw{do} 
\ECOM{left-to-right scan of $b$ and $S$, but random access to $S$}\\
14 \> \> \> \Keyw{if} $b[j] = 1$ \Keyw{then} \\
15 \> \> \> \> \Keyw{if} $j\not=0$ \Keyw{and} $b[j\!-\!1] = 1$ \Keyw{then} \\
16 \> \> \> \> \> $\ell \leftarrow \ell-1$ \ECOM{at least $\ell-1$ characters 
match by Lemma \ref{lem-PLCP entries}}\\
17 \> \> \> \> \Keyw{else} \\
18 \> \> \> \> \> $\ell \leftarrow m+1$ \ECOM{at least $m+1$ characters 
match by phase 1}\\
\\
19 \> \> \> \> \Keyw{if} $\Phi[j_I] \not= \bot$ \Keyw{then} \ECOM{$\PLCP[j_I]$
is irreducible}\\
20 \> \> \> \> \> \Keyw{while} $S[j+\ell] = S[\Phi[j_I]+\ell]$ \Keyw{do} \\
21 \> \> \> \> \> \> $\ell \leftarrow \ell+1$ \\
22 \> \> \> \> $\PLCP[j_I] \leftarrow \ell$ \ECOM{if $\PLCP[j_I]$ is reducible, 
no character comparison was needed}\\
23 \> \> \> \> $j_I \leftarrow j_I+1$ \\
\\
24 \> \> \Keyw{for} $i\leftarrow 0$ \Keyw{to} $n-1$ \Keyw{do} \ECOM{stream 
$\SUF$ and $\LCP$ from disk}  \\
25 \> \> \> \Keyw{if} $\LCP[i] > m$ \Keyw{then}  \\
26 \> \> \> \> $\LCP[i] \leftarrow \PLCP[rank_1(\SUF[i])]$ 
\end{tabbing}
\end{algorithm}

Algorithm \ref{alg-second phase} uses a bit array $b$, where $b[\SUF[i]] = 0$
if $\LCP[i]$ is known already 
(i.e., $b[j] = 0$ if $\PLCP[j]$ is known)
and $b[\SUF[i]] = 1$ if $\LCP[i]$ still must be computed 
(i.e., $b[j] = 1$ if $\PLCP[j]$ is unknown); see lines 1--4 of the algorithm.
In contrast to the $\Phi$-algorithm \cite{KAR:MAN:PUG:2009}, our algorithm
does not compute the whole $\Phi$-array ($\PLCP$-array, respectively) but only 
the $\nn$ many entries for which the lcp-value is still unknown (line 5). 
So if we would delete the values $\Phi[j]$ ($\PLCP[j]$, respectively) for 
which $b[j] = 0$ from the original $\Phi$-array ($\PLCP$-array, respectively) 
\cite{KAR:MAN:PUG:2009}, we would obtain our array $\Phi[0..\nn-1]$ 
($\PLCP[0..\nn-1]$, respectively).
We achieve a direct computation of $\Phi[0..\nn-1]$
with the help of a rank data structure for the bit array $b$ 
such that rank queries $rank_1(j)$ can be answered in constant time, where
$rank_1(j)$ returns the number of $1$'s up to position $j$ in $b$.
The for-loop in lines 7--9 fills our array $\Phi[0..\nn-1]$ but again there
is a difference to the original $\Phi$-array: reducible values are omitted!
After initialization of the counter $j_I$, the number $\ell$ (of characters
that can be skippped), and the $\PLCP$ array, 
the for-loop in lines 13--23 fills the array $\PLCP[0..\nn-1]$ by
scanning the $b$-array and the string $S$ from left to right. In line 14, 
the algorithm tests whether the lcp-value is still unknown 
(this is the case if $b[j] = 1$).
If so, it determines the number of characters that can be skippped
in lines 15--18. If $\PLCP[j_I]$ is irreducible (equivalently, 
$\Phi[j_I] \not= \bot$) then its correct values is computed by character
comparisons in lines 20--21. Otherwise, $\PLCP[j_I]$ is reducible and
$\PLCP[j_I]= \PLCP[j_I-1]-1$ by Lemma \ref{lem-PLCP entries}. In both cases
$\PLCP[j_I]$ is assigned the correct value in line 22. Finally, 
the missing entries in the $\LCP$-array (lcp-values in suffix array order)
are filled with the help of $\PLCP$-array (lcp-values in text order) in
lines 24--26.

Clearly, the first phase of our algorithm has a linear worst-case time
complexity. The same is true of the second phase as explained above.
Thus, the whole algorithm has a linear run-time. 

\bibliographystyle{plain}
\bibliography{diss_gog}

\Comment{
\section{Experimental results}
\label{sec-Experimental results}

\newcommand{\lcpKasai}{{\sf KLAAP}}
\newcommand{\lcpgo}{{\sf GO}}
\newcommand{\lcpgoII}{{\sf GO2}}
\newcommand{\lcpphi}{\ensuremath{\Phi}{\sf-KMP}}

\newcommand{\simpleI}{{\sf simple1}}
\newcommand{\simpleII}{{\sf simple2}}

Besides Algorithms \ref{alg-first phase} and \ref{alg-second phase},
we have implemented and compared four LACAs:
\begin{itemize}
\item[\lcpKasai:] A semi-external version of the algorithm presented in 
\cite{KAS:2001} using $5n$ bytes. Pseudo code is presented in the Appendix. 
\item[\lcpphi:] The semi-external version of the algorithm presented in 
\cite{KAR:MAN:PUG:2009} using $5n$ bytes. Pseudo code is presented in the 
Appendix. 
\item[\lcpgo:]	The hybrid algorithm presented in 
Section \ref{sec-The hybrid algorithm} using $2n$ bytes of memory.
\item[\lcpgoII:] The hybrid algorithm using buffered queues in the first phase 
and therefore usually using only $n$ bytes of memory (recall that---for space 
reasons---the algorithm is not explained in this paper).
\end{itemize}
The implementations are part of the first author's {\it sdsl}-library, 
which is available under the GPL licence at \url{http://goo.gl/FAEU}. 
This site also provides the complete set of test results. 
We used the {\it Pizza\&Chili} Corpus as test set. 
The tests were performed on a PC equipped with a Dual-Core AMD Opteron 1222 
processor and $4$GB of main memory. Tables \ref{tab-Running times} and 
\ref{tab-Memory consumption} show the running times (we meassured real time) 
and peak memory consumption for different text categories, text sizes, and 
algorithms. Of course, it would be interesting to know how the LACAs perform
relative to the other phases in the CST construction; our experimental
results in the Appendix answer this question.
Although the hybrid algorithm \lcpgo\
performs many more character comparisons than \lcpKasai\ or \lcpphi,
it clearly beats these algorithms. For example,
\lcpKasai\ or \lcpphi\ take  $5.2\cdot 10^{7}$ character comparisons 
for the input {\tt dblp.xml.50MB} while \lcpgo\ takes $6.4\cdot 10^9$. 
By contrast, \lcpgo\ requires only $1.8\cdot 10^8$ random acesses to the 
text, while the other algorithms need about $5.2\cdot 10^8$. 
The cost of random accesses increases with file size. 
While \lcpgo\ is two times faster than \lcpKasai\ for $20$MB of
DNA sequence, it is three time faster for $200$MB. 
Finally, let us focus on the two phases of \lcpgo. 
The first phase usually takes most space and time.
The parameter $\nn$ in the second phase lies in the range from $0.0001n$ 
(XML) to $0.24n$ (English text with many repetitions).
For small $\nn$ ($\nn<0.03n$), Algorithm \ref{alg-first phase} beats 
\lcpgo\ in time, while it is totally inappropriate for large $\nn$: 
it takes $60$ times longer for {\tt English.50MB} than \lcpgo. 
Algorithm \ref{alg-second phase}
never beats another algorithm in time or space.

\begin{table}
\begin{center}
\begin{tabular}{r*{4}{|*{4}{|r}}| }
 & \multicolumn{4}{|c}{20MB} &\multicolumn{4}{|c}{50MB} 
 & \multicolumn{4}{|c}{100MB} &\multicolumn{4}{|c}{200MB} 
 \\ \cline{2-17}
 &
            \begin{sideways} \lcpKasai \end{sideways} & \begin{sideways} \lcpphi \end{sideways}  & \begin{sideways} \lcpgo \end{sideways} & \begin{sideways} \lcpgoII \end{sideways} &
            \begin{sideways} \lcpKasai \end{sideways} & \begin{sideways} \lcpphi \end{sideways}  & \begin{sideways} \lcpgo \end{sideways} & \begin{sideways} \lcpgoII \end{sideways} &
            \begin{sideways} \lcpKasai \end{sideways} & \begin{sideways} \lcpphi \end{sideways}  & \begin{sideways} \lcpgo \end{sideways} & \begin{sideways} \lcpgoII \end{sideways} &
            \begin{sideways} \lcpKasai \end{sideways} & \begin{sideways} \lcpphi \end{sideways}  & \begin{sideways} \lcpgo \end{sideways} & \begin{sideways} \lcpgoII \end{sideways} 
\\ \hline
dna&6.5&6.7&3&4.3&18.9&17.6&7.9&13.1&46.4&38.1&16.1&23.9&111.4&89.4&35.7&51.7\\\hline 
English&5.9&4.8&7.3&9.9&16.2&13.7&15.5&20.5&37.6&35.2&29&36.5&101.4&76.4&63.7&81.1\\\hline 
dblp.xml&4.9&4.4&3.4&5.9&13.7&14.3&7.7&13.3&29.2&25.9&15.6&24.1&78.8&64.4&34.5&51.6\\\hline 
sources&4.5&4.2&3.1&6.7&13.6&11.6&9.2&14.1&28.2&26.8&18.2&29.7&77.9&53.6&41.8&62\\\hline 
proteins&5.4&4.7&5.1&6.4&19.7&14.7&11.7&15.8&40.2&36.2&24.5&33.1&100.1&75.7&60.8&76.5\\\hline 
\end{tabular}
\end{center}
\caption{Running times in seconds of the LACAs for test cases of the 
{\it Pizza\&Chili} Corpus. The first column contains the text category and 
the first row the size of the text.\label{tab-Running times}}
\end{table}

\begin{table}
\begin{center}
\begin{tabular}{r*{4}{|*{4}{|r}}| }
 & \multicolumn{4}{|c}{20MB} &\multicolumn{4}{|c}{50MB} 
 & \multicolumn{4}{|c}{100MB} &\multicolumn{4}{|c}{200MB} 
 
 \\ \cline{2-17}
 &
            \begin{sideways} \lcpKasai \end{sideways} & \begin{sideways} \lcpphi \end{sideways}  & \begin{sideways} \lcpgo \end{sideways} & \begin{sideways} \lcpgoII \end{sideways} &
            \begin{sideways} \lcpKasai \end{sideways} & \begin{sideways} \lcpphi \end{sideways}  & \begin{sideways} \lcpgo \end{sideways} & \begin{sideways} \lcpgoII \end{sideways} &
            \begin{sideways} \lcpKasai \end{sideways} & \begin{sideways} \lcpphi \end{sideways}  & \begin{sideways} \lcpgo \end{sideways} & \begin{sideways} \lcpgoII \end{sideways} &
            \begin{sideways} \lcpKasai \end{sideways} & \begin{sideways} \lcpphi \end{sideways}  & \begin{sideways} \lcpgo \end{sideways} & \begin{sideways} \lcpgoII \end{sideways} 
\\ \hline
dna&5.1&5.1&2.3&2&4.7&4.7&3.1&1.4&4.8&4.8&2.1&1.4&4.9&4.9&2.2&1.4\\\hline 
English&5.7&5.7&2.7&2.5&5.6&5.6&3.4&2.4&5&5&2.6&1.7&4.9&5.1&2.4&1.4\\\hline 
dblp.xml&5.5&5.5&2.7&2.4&5.6&5.6&3.3&2.3&5&5&2.6&1.6&4.8&5&2.5&1.3\\\hline 
sources&5.6&5.6&2.9&2.5&5.7&5.7&3.4&2.4&5&5&2.7&1.7&4.9&4.9&2.6&1.4\\\hline 
proteins&5.4&5.4&2.4&2.4&4.8&4.8&3.2&1.5&4.8&4.8&2.5&1.5&4.5&4.5&2.2&1.1\\\hline 
\end{tabular}
\end{center}
\caption{Memory in bytes per input character of the LACAs for test cases 
of the {\it Pizza\&Chili} Corpus. The first column contains the text category 
and the first row the size of the text.\label{tab-Memory consumption}}
\end{table}

\newpage

\newpage
\section*{Appendix}

\begin{figure}
	\begin{center}
	\scalebox{0.8}{
	}
	\end{center}
\caption{Four memory profiles of CST constructions for the input
{\tt dna.200MB}. The construction algorithms solely differ in the LACA
they use (from left to right): \lcpKasai, \lcpphi, \lcpgo, and \lcpgoII. 
Each profile shows the time spent on different tasks in the CST construction.
We will explain the leftmost profile in detail: The construction of the 
suffix array (sa) takes $62.1$ seconds. After that the Burrows and Wheeler 
Transform (bwt) is computed in $32.4$ seconds. The construction of the 
Wavelet Tree (wt) takes another $9.4$ seconds. Adding some samples of the 
suffix array yields the compressed suffix array after $115.2$ seconds.
The \lcpKasai\ algorithm first calculates the inverse suffix array (isa) in
$35$ seconds and then the $\LCP$-array (lcp), requiring $111.4$ seconds in 
total.
Finally, the topology of the suffix tree, which is represented by a sequence 
of balanced parentheses (bps), is calculated in $9.6$ seconds. 
The two right profiles impressivly show how space efficient 
the two new algorithms are. While they take about $2n$ or $n$ bytes 
in the first phase of the $\LCP$-array construction, they
usually take only a fraction of $n$ bytes in the second phase.
\label{fig-cst profile dna200}}	
\end{figure}

\begin{algorithm}
\caption{Kasai et al.'s LACA ($5n$ byte semi-external version).
\label{lcp_kasai} }
\begin{tabbing}
\quad \=\quad \=\quad \=\quad \=\quad \=\quad \=\quad \=\quad \=\quad \=\quad \=\quad \=\quad \kill
01 \> \> \Keyw{for} $i \leftarrow 0$ \Keyw{to} $n-1$ \Keyw{do} $\InverseSUF[\SUF[i]] \leftarrow i$ // stream $\SUF$ from disk \\
   \> \> // store $\InverseSUF$ to disk \\
02 \> \> $\ell \leftarrow 0$ \\
03 \> \> \Keyw{for} $i \leftarrow 0$ \Keyw{to} $n-1$ \Keyw{do} \\
04 \> \> \> $\InverseSUF_i \leftarrow \InverseSUF[i]$ // stream $\InverseSUF$ from disk \\
05 \> \> \> \Keyw{if} $\InverseSUF_i>0$ \Keyw{do}   \\
06 \> \> \> \> $j \leftarrow \SUF[\InverseSUF_i-1]$ \\
07 \> \> \> \> \Keyw{while} $S[i+\ell] = S[j+\ell]$ \Keyw{do} \\
08 \> \> \> \> \> $\ell \leftarrow \ell+1$ \\
09 \> \> \> \> $\SUF[\InverseSUF_i-1] \leftarrow \ell$ // overwrite $\SUF[\InverseSUF_i-1]$ with $\LCP[\InverseSUF_i]$ \\
10 \> \> \> $\ell \leftarrow max(\ell-1,0)$ \\  
11 \> \> $\SUF[n-1] \leftarrow -1$ \\ 
12 \> \> // store overwritten $\SUF$ as result to disk and shift indices by $1$  \\
\end{tabbing}
\end{algorithm}

\begin{algorithm}
\caption{K\"arkk\"ainen et al.'s LACA ($5n$ byte semi-external version).
\label{lcp_phi}
}
\begin{tabbing}
\quad \=\quad \=\quad \=\quad \=\quad \=\quad \=\quad \=\quad \=\quad \=\quad \=\quad \=\quad \kill
01 \> \> \Keyw{for} $i \leftarrow 0$ \Keyw{to} $n-1$ \Keyw{do} $\Phi[\SUF[i]] \leftarrow \SUF[i-1]$ // stream $\SUF$ from disk  \\
02 \> \> $\ell \leftarrow 0$ \\
03 \> \> \Keyw{for} $i \leftarrow 0$ \Keyw{to} $n-1$ \Keyw{do} \\
04 \> \> \> \Keyw{while} $S[i+\ell] = S[\Phi[i]+\ell]$ \Keyw{do} \\
05 \> \> \> \> $\ell \leftarrow \ell+1$   \\
06 \> \> \> $\PLCP[i] \leftarrow \ell$  \\
07 \> \> \> $\ell \leftarrow max(\ell-1,0)$ \\
08 \> \> \Keyw{for} $i \leftarrow 0$ \Keyw{to} $n-1$ \Keyw{do} \\
09 \> \> \>  $\LCP[i] \leftarrow \PLCP[\SUF[i]]$ // stream $\SUF$ from disk, write $\LCP$ buffered to disk
\end{tabbing}
\end{algorithm}
}

\Comment{
\begin{table}
\begin{tabular}{r*{2}{|*{3}{|c}} }
 & \multicolumn{3}{|c}{character comparisons} &\multicolumn{3}{|c}{random accesses to text}  \\
            &\begin{sideways}\ \lcpKasai \end{sideways} & \begin{sideways}\ \lcpphi \end{sideways}  & \begin{sideways}\ \lcpgo \end{sideways}  &
           \begin{sideways}\ \lcpKasai \end{sideways} & \begin{sideways}\ \lcpphi \end{sideways}  & \begin{sideways}\ \lcpgo \end{sideways}  \\
DNA.50MB & 		$5.2\cdot 10^8$ & $5.2\cdot 10^8$ & 	 $2.5\cdot 10^9$ & 	 $5.2\cdot 10^8$ & 	 $5.2\cdot 10^8$ & 	 $2.0\cdot 10^8$ \\ 	 	   
English.50MB & 		$5.2\cdot 10^8$ & $5.2\cdot 10^8$ & 	 $1.2\cdot 10^{10}$ & 	 $5.2\cdot 10^8$ & 	 $5.2\cdot 10^8$ & 	 $2.0\cdot 10^8$ \\ 	 	   
dblp.xml.50MB & $5.2\cdot 10^8$ & $5.2\cdot 10^8$ & 	 $6.4\cdot 10^9$ & 	 $5.2\cdot 10^8$ & 	 $5.2\cdot 10^8$ & 	 $1.6\cdot 10^8$ \\ 	 	   
sources.50MB & $5.2\cdot 10^8$ & $5.2\cdot 10^8$ & 	 $7.7\cdot 10^9$ & 	 $5.2\cdot 10^8$ & 	 $5.2\cdot 10^8$ & 	 $1.7\cdot 10^8$ \\ 	 	   
proteins.50MB & $5.2\cdot 10^8$ & $5.2\cdot 10^8$ & 	 $7.9\cdot 10^9$ & 	 $5.2\cdot 10^8$ & 	 $5.2\cdot 10^8$ & 	 $2.3\cdot 10^8$ \\ 	 	   
\end{tabular}
\caption{Number of character comparisons and random accesses to the text during the execution of the algorithms.\label{tab-Character comparisons}}
\end{table}
} 

\end{document}